\documentclass[12pt, draftclsnofoot, onecolumn]{IEEEtran}
\usepackage{graphics,graphicx}
\usepackage{amssymb,amsmath}
\usepackage{amsmath,empheq}
\usepackage{url}
\usepackage[utf8]{inputenc}
\usepackage{algorithm}
\usepackage{comment}
\usepackage{algorithmic}
\usepackage{hyperref}
\usepackage{indentfirst}
\usepackage{setspace}
\usepackage{cases}

\usepackage{booktabs}
\usepackage{multicol}
\usepackage{cite}
\usepackage{amsthm}

\newtheorem{theorem}{Theorem}
\newtheorem{corollary}{Corollary}

\newtheorem{lemma}{Lemma}

\theoremstyle{remark}
\newtheorem{remark}{Remark}

\usepackage{stfloats}

\usepackage{setspace}

\usepackage{suffix}
\usepackage{mathtools}

\DeclarePairedDelimiterX\MeijerM[3]{\lparen}{\rparen}%
{\begin{smallmatrix}#1 \\ #2\end{smallmatrix}\delimsize\vert\,#3}

\newcommand\MeijerG[8][]{%
  G^{\,#2,#3}_{#4,#5}\MeijerM[#1]{#6}{#7}{#8}}

\WithSuffix\newcommand\MeijerG*[7]{%
  G^{\,#1,#2}_{#3,#4}\MeijerM*{#5}{#6}{#7}}

\usepackage{tikz}

\usepackage{booktabs}
\usepackage[nolist]{acronym}

\usepackage{tikz}
\definecolor{gold}{rgb}{0.85,.66,0}
\definecolor{cian}{rgb}{.02,.7,.95}
\definecolor{ppp}{rgb}{.7,.3,.82}

\begin{document}
\title{RIS-aided Cooperative FD-SWIPT-NOMA Outage Performance in Nakagami-$m$ Channels}
\author{{Wilson de Souza Junior} and {Taufik Abrão}

\thanks{This work was supported in part by the CAPES (Financial Code 001) and National Council for Scientific and Technological Development (CNPq) of Brazil under Grant 310681/2019-7.}
\thanks{W. Junior and T. Abrão are with the Department of  Electrical Engineering (DEEL). State University of Londrina  (UEL). Po.Box 10.011, CEP:86057-970, Londrina, PR, Brazil.  Email: wilsoonjr98@gmail.com; taufik@uel.br}
}
\maketitle

\begin{abstract}
In this work, we investigate reconfigurable intelligent surfaces (RIS)-assisted cooperative non-orthogonal multiple access (C-NOMA) consisting of two paired users, where the phases of RIS are configured to boost the cell-center device. The cell-center device is designated to act as a full-duplex (FD) relay to assist the cell-edge device. The cell-center device does not use its battery energy to cooperate but harvests energy using simultaneous wireless information power transfer (SWIPT). 
A more practical non-linear energy harvesting model is considered. Expressions for outage probability (OP) and ergodic rate (ER) are devised, assuming that all users' links undergo Nakagami-$m$ channel fading. We first approximate the harvested power as Gamma random variables via the moments matching technique. This allows us to derive analytical OP/ER expressions that are simple to compute yet accurate for a wide range of RIS passive elements configurations, energy harvesting (EH) coefficients, and residual self-interference (SI) levels, being extensively validated by numerical simulations. The OP expressions reveal how paramount is to mitigate the SI in the FD relay mode since for reasonable values of residual SI coefficient ($\omega \geq -20$dB), it is notable its detrimental effect on the system performance. Also, numerical results reveal that by increasing the number of RIS elements can be much more beneficial to the cooperative system than the non-cooperative system.
\end{abstract}


\section{Introduction}
Among several highly important services addressed by the fifth generation (5G) communication systems, and also predicted to be present in wireless communications systems beyond 5G, massive machine-type communication (mMTC) is a use case scenario widely targeted due to the exponential increasing of devices interconnected inside a network. Since wireless networks have become denser, one of the challenges is to support heavy traffic. In recent studies in the literature, it has been proven that {\it non-orthogonal multiple access} (NOMA) can be superior to the multiple access schemes as deployed in the past generations of cellular networks, including frequency division multiple access (FDMA), time division multiple access (TDMA), and code division multiple access (CDMA), which possibly may not be able to scale to meet such new 5G use case demands. NOMA can be regarded as an interesting candidate to overcome such challenges, once NOMA technology allows multiple users to transmit in the same resource block (RB), which can lead to higher spectral efficiency (SE) and energy efficiency (EE), than the usual {\it orthogonal multiple access} (OMA). By adopting NOMA, successive interference cancellation (SIC) is paramount to distinguish the respective user's signal, being essential for NOMA to work suitably.

In a cellular system, the connectivity of the cell-edge devices is often affected due to their geographical position; hence, to improve fairness, the cell-edge devices need to be allocated a large number of resources, which can harm the QoS of the cell-center users. Cooperative communications can be a potential solution for this issue once it can improve the cell-edge users' data rate and enhance the network's fairness.  The integration of NOMA and cooperative user relaying has attracted significant attention mainly due to the natural matching of both techniques since the cell-edge user information can be known to the cell-center users \cite{zeng2020cooperative}. Such a scenario is critical since device-to-device (D2D) communication is one of the key technologies for the next generations of communication systems \cite{jiang}, \cite{mahmood}. Cooperative communications can be performed under two distinct modes: half-duplex (HD) and full-duplex (FD). The HD mode is known for sub-dividing the transmission time block, which can degrade the SE, while the FD mode can be performed simultaneously at the cost of \textit{self-interference} (SI) aggregation. Notably, the cooperative NOMA (C-NOMA) system was first proposed and studied in \cite{7117391}, where the cell-center users are selected as relays to guarantee the QoS of the cell-edge users; thus, the OP and the diversity order are performance metrics of interest. In addition, \cite{8026173,zhang2016full} investigated the performance of FD/HD cooperative NOMA systems.

Although performance gains could be achieved with cooperative communications, the deployment of such technology can be detrimental to the battery lifetime of the devices since inevitably, it drains the battery energy of the cell-center users. In this sense,  {\it simultaneous wireless information and power transfer} (SWIPT) is a promising and sustainable technology that can be useful to address this issue while enabling the implementation of the internet of things (IoT), since IoT devices are usually energy-limited for relay cooperation. Moreover, the SWIPT technique allows the devices to harvest energy from the ambient radio-frequency (RF) sources and reuse it to cooperate with the cell-edge users. SWIPT also can be implemented in two ways, {\it power-splitting} (PS) and {\it time-switching} (TS) \cite{ashraf}. The PS-SWIPT protocol splits the received signal power from the base station (BS) into two parts, one for {\it information decoding} (ID), and the other for {\it energy harvesting} (EH) purposes. On the other hand, the TS-SWIPT protocol necessarily splits  the time slots to perform the ID and EH processes separately. \cite{liu2022outage,liu2022system,tran2021swipt,trin} investigate the system performance of the SWIPT-assisted C-NOMA system over FD/HD mode, carrying out extensive analyses on the OP/ER. Moreover,  \cite{wu2019transceiver} proposes an alternative optimization-based algorithm to jointly optimize the power allocation, power splitting, receiver filter and transmit beamforming in a SWIPT-assisted C-NOMA system.

Recently, {\it reflecting intelligent surface} (RIS) has attracted remarkable research attention due to its capability of changing and customizing the wireless propagation environment, hence supporting high system throughput and being useful in indoor/outdoor scenarios where dense obstacles arise. The RIS is composed of scattering elements, {\it i.e.}, artificial meta-material structures composed of adaptive composite material layers, which can reflect incident electromagnetic waves and can be configured to increase the signal level in a specific direction for a priority user, and it is calling attention mainly due to its features: the capacity to be sustainable, low-power consumption, enhancing communications metrics, facility of implementation and installation, compatibility and low cost. In \cite{cheng2021downlink}, the RIS-assisted non-C-NOMA system has been analytically validated in terms of OP considering two-user scenarios; indeed, RIS-aided communication can substantially improve the system OP. In \cite{adaoRIS}, the impact of the BS-user direct link on the system performance is evaluated by comparing the RIS-aided system with the relay-aided system. Furthermore, the performance of RIS-aided non-C-NOMA system has been analyzed in \cite{YueRIS} from the perspective of imperfect SIC effects.
On the other hand, in \cite{peppasRIS}, the authors assess the impact of RIS phase shift design on the OP, ER, and bit error probability (BER) through two phase-shift configurations: random phase, and coherent phase-shifting under Nakagami-$m$ fading. Recently, \cite{LiRIS} provides valuable analytical results on the OP for downlink RIS-assisted backscatter communications with NOMA. In \cite{9269324}, the RIS-NOMA system OP under hardware impairments has been analyzed analytically: an accurate closed-form for the OP was developed. Very recently, the end-to-end channel statistics for both weakest and strongest users in a {\it non}-cooperative RIS-aided NOMA system under Nakagami-$m$ was derived in \cite{tahir}. In \cite{zuoCoopRIS}, a RIS-aided cooperative NOMA scheme was analyzed from the perspective of power consumption.

\subsection{Motivation and Contributions}

To further improve the throughput, reliability, and fairness of mobile devices, the integration of different technologies such as NOMA, cooperative system with SWIPT and RIS is very promising; furthermore, there are few works in the literature dealing with RIS-aided C-NOMA SWIPT systems \cite{9926196,9771850,elhattab2021reconfigurable}. A hybrid TS and PW EH relaying for RIS-NOMA system with transmit antenna selection is proposed in \cite{9926196}, while in \cite{elhattab2021reconfigurable} the authors intended to minimize the transmit power at both the BS and at the user-cooperating relay. Differently, in \cite{9771850}, the authors proposed an algorithm to jointly optimize the beamforming and the power splitting coefficient. However, to the best of our knowledge, so far, no studies have provided a solid investigation on the OP/ER performance of RIS-aided NOMA systems, the integration of RIS-aided cooperative communications assuming non-linear energy harvesting model operating under generalized Nakagami-$m$ fading channels.

Motivated by the aforementioned facts, in this work we aim to investigate the {\it potential benefits of combining RIS, NOMA, and cooperative communications with SWIPT and non-linear EH circuits}. Herein, in order to suitably unveil the system OP performance, we focus on a two-user scenario with perfect SIC and perfect knowing channel state information (CSI) at the BS\footnote{Most complex, intricate scenarios, and performance metrics such as secrecy outage, sum rate, etc, are out of the scope of this work and will be treated in future works.}.

\vspace{2mm}
In light of the above motivations and challenges, the main \textbf{\textit{contributions}} of this work are threefold and can be summarized as follows:
\begin{itemize}
\item We analyze a cooperative NOMA system scenario under Nakagami-$m$ fading channels, combined with the SWIPT technique adopting a non-linear energy harvesting model, where the cell-center device cooperates with the cell-edge device without harming itself in terms of battery lifetime. Besides, we investigate the use of RIS aiming to identify its benefits and drawbacks operating under the considered scenario. 
\item  We derive novel general expressions for the OP and upper bound of ER for the analyzed system. Since such expression is parameterized w.r.t. the system and channel parameters, the effect of each parameter can be effectively scanned as a function of the number of RIS elements. 
\item  Comprehensive numerical results simulations for the  OP, ER and user rate corroborating the effectiveness and accuracy of the proposed analytical performance expressions.

\end{itemize}
The adopted methodology allows us to assess analytically the RIS-aided cooperative SWIPT-NOMA performance, which is essential to predict how the system acts in practice. 

\vspace{2mm}
\noindent\textbf{\textit{Notation}}:  
$\Gamma(\cdot)$ is the gamma function; $\gamma(\cdot,\cdot)$ is the lower incomplete gamma function; $\Gamma(\cdot,\cdot)$ is the upper incomplete gamma function; ${\rm Ei}(\cdot)$ is exponential integral function;  $X \sim \mathcal{CN}(\mu,\sigma^2)$ denotes a random variable $X$ following a Complex Normal distribution with mean $\mu$ and variance $\sigma^2$; $X \sim \text{Gamma}(k,\theta)$ denotes a random variable following a Gamma distribution with shape parameter $k$ and scale parameter $\delta$; $X \sim \text{Exponential}(\lambda)$ denotes a random variable following a Exponential distribution with rate parameter $\lambda$; $X \sim \text{Rayleigh}(\sigma)$ denotes a random variable following a Rayleigh distribution with scale parameter $\sigma$; the magnitude of a complex number $z$ is expressed by $|z|$; $\text{arg}(\cdot)$ denotes the argument of a complex number; $\text{diag}(\cdot)$ denotes the diagonal operator; vectors and matrices are represented by bold-face letters; $F_{X}$ denotes the cumulative density function (CDF) of $X$; $f_{X}$ denotes the probability density function (PDF); and ${\rm Pr}(\cdot)$ expresses probability.

\section{System Model}\label{sec:sys_model}
Let us consider a RIS-assisted C-NOMA downlink where a source ($S$) equipped with a single antenna simultaneously serves two devices. Let us denote the cell-center device as $D_1$, while $D_2$ is the cell-edge device as represented in Fig. \ref{fig:sys}. The transmission process is assisted by an $N$-elements RIS.
\begin{figure}[!htbp]
\centering
\includegraphics[trim={1cm 10cm 2cm 10cm},clip,width=.75\linewidth]{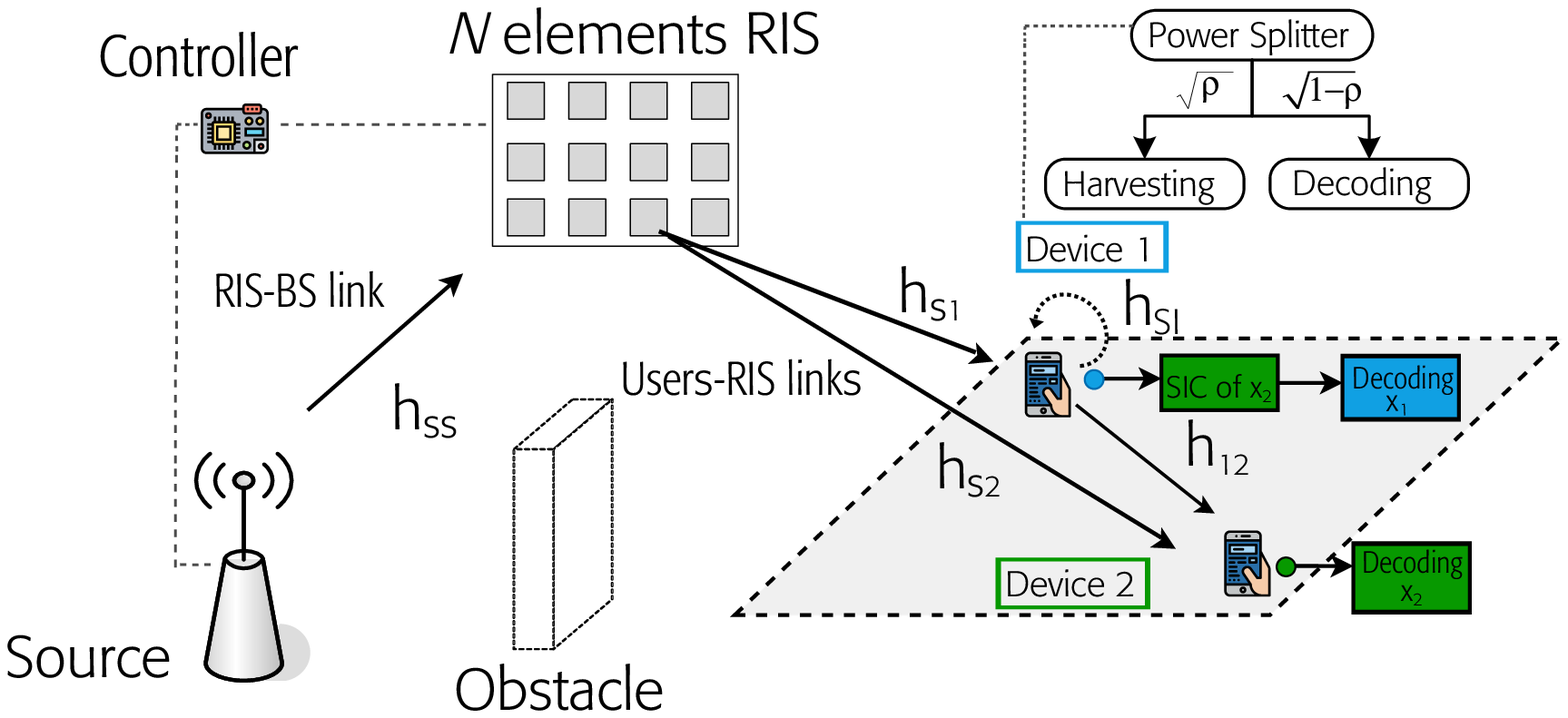}
\vspace{-6mm}
\caption{Downlink C-FD-SWIPT-NOMA system model, consisting of a source ($S$) and two paired devices, $D_1$ and $D_2$. To improve the performance of $D_2$, the device $D_1$ acts as a full-duplex relay under the self-interference (SI) effect.}
\label{fig:sys}
\end{figure}

The direct link between the source and the devices is assumed to be completely obstructed. The application scenarios for the system model illustrated in Fig. \ref{fig:sys} are leveraged with the advent of IoT systems. The adopted setup directly applies to environments where D2D communication can benefit, including offices and residences. Furthermore, there is a vast potential for application in industries and
communication/automation between industrial robots.

To ensure QoS to the system, $D_1$ can act as a cooperative relay adopting the decode-and-forward (DF) protocol \cite{wan2017cooperativeDF}, \cite{Weid2d}. In order to reach ultra-low latency in the cooperative framework, a primordial feature in the 5G and 6G communication systems, in this work we have assumed that the cooperative device ($D_1$) is equipped with two antennas, where the first one is responsible for receiving the signal and the other antenna to relaying \cite{alsaba2018full}; hence, the cooperative process (at the user-relay) can occur in an FD mode \cite{huang2019energy,liau2018cooperative,tran2021swipt,wu2019transceiver,zeng2020cooperative} but subject to self-interference (SI). SI or loop-interference (LI) refers to the signal that is transmitted by the transmitter antenna operating in the FD node and looped back to the receiver antenna at the same node \cite{zhang2016full}.

To do not jeopardize the battery lifetime of $D_1$, such a device can take advantage of the SWIPT technique by adopting the PS architecture \cite{khaledPS,trin}, where a fraction $\rho$ of the received power is utilized to energy harvesting and the remaining fraction $(1-\rho)$ to perform information decoding. The energy harvested can be fully utilized to relay the rebuilt message version of $D_2$.

The RIS operation is encapsulated by its phase shift matrix, which is assumed to be ideally a lossless surface and given as
\begin{equation}
\boldsymbol{\Phi} = {\rm{diag}}( e^{j\phi_1}, e^{j\phi_2},\dots, e^{j\phi_N})
\end{equation}
where $\phi_n \in (0,2\pi]$ is the phase shift applied to the $n$-th element of the RIS. We consider that the RIS is programmed by a dedicated controller connected to the $S$ via a high-speed backhaul \cite{do2021multi}, whose main aim is to update systematically the RIS phases shift at each coherence-time once the phase-shift variables are directly related to the CSI. 

We assume that all link channels experience quasi-static flat fading over a coherence time and vary independently from one coherence time to another. The CSI is assumed to be perfectly known at the $S$. The small-scale channel from $S$ to the RIS is denoted by $\mathbf{h}_{ss}$ $\in \mathbb{C}^{N}$, while the small-scale channel between the RIS and the $\ell$-th device is denoted by $\mathbf{h}_{s\ell}$ $\in$ $\mathbb{C}^{N}$ $\forall \ell \in \{1,2\}$. It is assumed that  $|[\mathbf{h}_{ss}]_n| = |h_{ss,n}|$ and  $|[\mathbf{h}_{s\ell}]_n| = |h_{s\ell,n}|$ are independent r.v. following Nakagami-$m$ fading, $i.e.$, 
\begin{IEEEeqnarray}{rCl}
    &&|h_{ss,n}| \sim {\rm Nakagami}(m_{ss},\Omega_{ss}), \quad \forall n \in \{1,\dots,N\} \nonumber 
    \\
    &&|h_{s\ell,n}| \sim {\rm Nakagami}(m_{s\ell},\Omega_{s\ell}), \quad \forall \ell \in \{1,2\}, 
\end{IEEEeqnarray}
while the phases are uniformly distributed in the $[0,2\pi]$ range. Hence, the equivalent channel for the $\ell$-th device can be written as:
\begin{equation} 
h_\ell = \sqrt{\beta_{ss} \beta_{s\ell}}\mathbf{h}_{ss}^H \mathbf{\Theta} \mathbf{h}_{s\ell},
\end{equation}
where $\beta_{ss}$ and $\beta_{s\ell}$ are the large-scale fading of $S$ $\rightarrow$ RIS link and RIS $\rightarrow$ $D_\ell$ link respectively.
Besides, for the $D_1$ $\rightarrow$ $D_2$ communication link (D2D), the 
$h_{12}$ follows a complex Gaussian distribution, {\it i.e.}, 
$h_{12} \sim \mathcal{CN}(0,\sqrt{\beta_{12}})$ $\forall n \in \{1,\dots,N\}$, where $\beta_{12}$ is the path loss of D2D communication.

In this work, we consider the RIS serving the user-relay with a coherent combination; thus, the phases shift of RIS are set to
\begin{IEEEeqnarray}{rCl}
    \phi_n = -\angle h_{ss,n} - \angle h_{s1,n}, \qquad &&\forall n=\{1,\dots,N\}
\end{IEEEeqnarray}

Therefore, the cascaded channel can be written as
\begin{IEEEeqnarray}{rCl} \label{eq:eqh}
    h_\ell &=& \sqrt{\beta_{ss} \beta_{s\ell}} \mathbf{h}_{ss} \boldsymbol{\Phi} \mathbf{h}_{s\ell}, \qquad \forall \ell \in \{1,2\} 
\end{IEEEeqnarray}
with $\mathbf{h}_{ss} \boldsymbol{\Phi} \mathbf{h}_{s1}=\sum_{n=1}^{N} |h_{ss,n}| |h_{s1,n}|$ being a real-number\footnote{Sum of product between two Nakagami-$m$ r.v., since the RIS is programmed to cancel the phase of $h_{ss,n}$ and $h_{s1,n}$} and $\mathbf{h}_{ss} \boldsymbol{\Phi} \mathbf{h}_{s2} = \sum_{n=1}^{N}e^{j\phi_n}  h_{ss,n} h_{s2,n}$, a complex number once the phases of RIS appear random for $D_2$.

\subsection{Signal Model}
In NOMA, $S$ transmits a superimposed signal $x(t)$ which propagates in direction to the devices through the RIS, with $x = \sqrt{\alpha_1}x_1 + \sqrt{\alpha_2}x_2$, where $\alpha_\ell \in (0,1)$ denotes the power allocation coefficients, with $\alpha_1 + \alpha_2 = 1$, $\alpha_2 > \alpha_1$ and $x_\ell$ with $\mathbb{E}[|x_\ell|^2] = 1$ is the message of $\ell$-th device, $\ell \in \{1,2\}$.

\subsubsection{Device 1} The observation at the $D_1$ which will be designated for ID can be written as follows
\begin{IEEEeqnarray}{rCl} \label{eq:receivedstrong}
\hspace{-5mm}{\rm{y}}_1^{\rm{ID}}(t) &=& \underbrace{h_1 \sqrt{P_t(1-\rho)} {x}(t)}_{ \substack{\text{Superimposed} \\ \text{information}}} + 
\underbrace{h_{SI} \sqrt{P_H} \hat{x}_2(t-\tau)}_{\text{Self-interference}} + \underbrace{n_1(t)}_{\text{noise}},
\end{IEEEeqnarray}
where $0\leq\rho\leq 1$ is the received power fraction utilized  to energy harvesting (EH factor), $P_{t}$ is the transmit power and the SI term comes from by adopting the FD mode at the user-relay. In this work, we consider that channel coefficient related to the SI, $h_{SI}$, undergo a zero mean complex-Normal distribution \cite{liu2022outage,liu2022system,jin2022secure,YueCoop2}, with power $|h_{SI}|^2=\omega$.  
Besides, $\hat{x}_2(t-\tau)$ is the retransmitted and rebuilt message of $D_2$ by $D_1$, $\tau$ denotes the processing delay at the user-relay caused by FD mode, which is assumed lower than the coherence time. The additive white Gaussian noise (AWGN) at the $D_1$ is modeled as $n_1 \sim \mathcal{CN}(0,\sigma^2_1)$.

\subsubsection{Device 2}
The observation at the $D_2$ can be written as follows
\begin{IEEEeqnarray}{rCl}
{\rm{y}}_2(t) = \underbrace{h_1 \sqrt{P_t} x(t)}_{\substack{\text{Superimposed} \\ \text{information}}}  + \underbrace{h_{12} \sqrt{P_{H}} \Hat{x}_2(t-\tau)}_{\substack{\text{Cooperative} \\ \text{transmission}}} + \underbrace{n_2(t)}_{\text{noise}},
\end{IEEEeqnarray}
where $n_2$ is the AWGN at the $D_2$. 

\vspace{2mm}
\noindent\textbf{\textit{Non-Linear EH Model}}. For the EH process, we employ a practical non-linear model \cite{boshkovska2015practical}; thus, the power deployed in the relaying step at $D_1$ ($P_H$) can be expressed as:
\begin{equation}\label{eq:EHmodel}
{P_{H}(P_{in}) = \frac{P_{th}\left( \frac{ 1}{1+e^{-a(\rho P_{in}-b)}} - \frac{1}{1+e^{ab}} \right)}{1-\frac{1}{1+e^{a b}}}},
\end{equation}
where $P_{th}$ is the threshold harvested power in saturation, $a$ and $b$ are constants related to the EH circuits as capacitance, resistance, and diode turn-on voltage. 
The adopted non-linear EH model from \cite{boshkovska2015practical} closely matches experimental/practical EH circuit results for both the low ($\mu$W) and high ($m$W) wireless power harvested regime. The RF input power in the EH circuit at $D_1$ is defined as:

\begin{equation}
    P_{in} = P_t  |h_1|^2.
\end{equation}

\subsection{Signal-Interference-to-Noise-Ratio}
$D_1$ receives a superimposed message ${\rm y}_1$ from the $S$ $\rightarrow$ RIS $\rightarrow$ $D_1$, so, according to NOMA the message of $D_2$ is detected fist, and the corresponding signal-to-interference-plus-noise (SINR) is given by
\begin{equation}
    {\rm SINR}_{D_1}^{x_2} = \frac{|h_1|^2 (1-\rho) P_t \alpha_2 }{ |h_1|^2  (1-\rho) P_t \alpha_1 +  |h_{SI}|^2P_H + \sigma^2},
\end{equation}
where we consider without loss of generalization that $\sigma_1^2 = \sigma_2^2 = \sigma^2$. After message $x_2$ is detected, it is eliminated from the received signal \eqref{eq:receivedstrong} by performing the SIC process\footnote{Here we consider that the SIC process is performed perfectly, $i.e.$, we do not take into account an eventual residual error from this process.}, thus, the SINR in $D_1$ for detecting $x_1$ is given by

\begin{equation} \label{eq:SINRd1x1}
    {\rm SINR}_{D_1}^{x_1} = \frac{|h_1|^2  (1-\rho) P_t \alpha_1 }{ |h_{SI}|^2 P_H + \sigma^2}.
\end{equation}

At $D_2$, the received SINR to detect $x_2$ from $S$ $\rightarrow$ RIS $\rightarrow$ $D_2$ link, and the received SNR to detect $x_2$ from $D_1$ $\rightarrow$ $D_2$ link are respectively expressed as
\begin{equation}
     {\rm SINR}_{D_2,S}^{x_2} = \frac{|h_2|^2 P_t \alpha_2 }{|h_2|^2 P_t \alpha_1 + \sigma^2},
\end{equation}
and
\begin{equation}
    {\rm SNR}_{D_2,D_1}^{x_2} =  \frac{P_{H} |h_{12}|^2}{\sigma^2}.
\end{equation}

By adopting the maximum ratio combining (MRC) rule, the overall SINR in $D_2$ is equivalent to the sum of SINR from $S$ $\rightarrow$ RIS $\rightarrow$ $D_2$ link and SNR from $D_1$ $\rightarrow$ $D_2$ link, which can be expressed as \cite{wu2019transceiver,liangcoop,YueCoop2}
\begin{IEEEeqnarray}{rCl}
{\rm SINR}_{D_2}^{x_2} &=& {\rm SINR}_{D_2,S}^{x_2} + {\rm SNR}_{D_2,D_1}^{x_2} \nonumber \\
&=&  \frac{|h_2|^2 P_t \alpha_2 }{|h_2|^2 P_t \alpha_1 + \sigma^2} + \frac{P_{H} |h_{12}|^2}{\sigma^2}.
\end{IEEEeqnarray}

\section{Statistics of channel and harvested power}

Our goal is to obtain new expressions that characterize the OP and ER for both devices in the RIS-aided C-NOMA-SWIPT. For that reason, in the following two-subsection we characterize statistically the cascaded channel for both devices as well as the harvested power in $D_1$ respectively.

\subsection{Statistical Channel Characterization}

Before proceeding to the OP derivation, it is paramount to characterize the channel statistically. Let us denote $X_1 =  \sum_{n=1}^{N} |h_{ss,n}| |h_{s1,n}| $ and $X_2 = \left| \sum_{n=1}^{N} e^{j\phi_n} h_{ss,n} h_{s2,n} \right|$. According to Lemma 2 of \cite{tahir2021outage}, the distribution of $X_1$ and $X_2$, can be approximated as 
\begin{IEEEeqnarray}{rCl}
    X_1 &\overset{\rm approx}{\sim}& {\rm Gamma}\left(N \frac{\mu_{ss}^2 \mu_{s1}^2}{1-\mu_{ss}^2 \mu_{s1}^2},\frac{{1-\mu_{ss}^2 \mu_{s1}^2}}{\mu_{ss} \mu_{s1}}\right),
    \\
    X_2 &\overset{\rm approx}{\sim}& {\rm Rayleigh}\left(\sqrt{\frac{N}{2}}\right),
\end{IEEEeqnarray}
where $\mu_{ss}$ and $\mu_{s1}$ are the mean of a Nakagami-$m$ r.v. given as \cite[Table 5.2]{papoulis1989probability}

\begin{equation}
    \mu_{ss} = \frac{\Gamma(m_{ss}+1/2)}{\Gamma(m_{ss})} \sqrt{\frac{\Omega_{ss}}{m_{ss}}}, \quad \mu_{s1} = \frac{\Gamma(m_{s1}+1/2)}{\Gamma(m_{s1})} \sqrt{\frac{\Omega_{s1}}{m_{s1}}}.
\end{equation}

Proceeding, Lemma 1 of \cite{tahir2021outage} states that the distribution of $X_1^2$ and $X_2^2$, can be approximated respectively as
\begin{IEEEeqnarray}{rCl}
    &X_1^2& \hspace{0.1cm} \overset{\rm approx}{\sim} {\rm Gamma}\left(k_1,\theta_1\right), \\
    &X_2^2& \hspace{0.1cm} \overset{\rm approx}{\sim} {\rm Exponential}\left(N\right),
\end{IEEEeqnarray}
where 
\begin{equation}
    k_1 = \frac{\left(\mu^{(2)}_{X_1}\right)^2}{\mu_{X_1}^{(4)} - \left(\mu^{(2)}_{X_1}\right)^2}, \quad \theta_1 = \frac{\mu_{X_1}^{(4)} - \left(\mu^{(2)}_{X_1}\right)^2}{\mu^{(2)}_{X_1}},
\end{equation}
with $\mu_{X_1}^{(m)}$ being the $m$-th moment of a Gamma r.v. given as
\begin{IEEEeqnarray}{rCl}
    \mu_{X
    _1}^{(m)} &=& \frac{\left(\frac{{1-\mu_{ss}^2 \mu_{s1}^2}}{\mu_{ss} \mu_{s1}}\right)^m\Gamma\left(N \frac{\mu_{ss}^2 \mu_{s1}^2}{1-\mu_{ss}^2 \mu_{s1}^2} + m\right)}{\Gamma\left(N \frac{\mu_{ss}^2 \mu_{s1}^2}{1-\mu_{ss}^2 \mu_{s1}^2}\right)},
\end{IEEEeqnarray}
thus, the PDF and CDF of $X_1^2$ and $X_2^2$ can be written respectively as

\begin{IEEEeqnarray}{rCl} \label{eq:normpdf}
    f_{X_1^2}(x) &=& \frac{x^{k1-1} e^{\frac{-x}{\theta_1}}}{\Gamma(k_1)\theta_1^{k_1}}, \quad  f_{X_2^2}(x) =\frac{1}{N}e^{-\frac{x}{N}}, \\
    F_{X_1^2}(x) &=& \frac{\gamma\left(k1,\frac{x}{\theta_1}\right)}{\Gamma(k_1)}, \quad F_{X_2^2}(x) = 1 - e^{-\frac{x}{N}}. \label{eq:normcdf}
\end{IEEEeqnarray}

\subsection{Statistical Harvested Power Characterization}
We also should analyze the distribution of the harvested power in $D_1$ when $D_1$ intends to act as a relay ($\rho\neq0$), thus for this purpose, the following lemma is conceived.

\begin{lemma} \label{lemma:PH}
Let $\zeta = aP_t\rho\beta_{ss}\beta_{s1}\theta_1$, the distribution of $P_H$ in Eq. \eqref{eq:EHmodel} can be approximated as a Gamma r.v. $P_H \overset{{\rm approx}}{\sim} {\rm Gamma}\left(k_{P_H},\frac{P_{th}}{1-\frac{1}{1+e^{ab}}}\theta_{P_H}\right)$ for $\zeta\ll 1$ where $k_{P_H}$ and $\theta_{P_H}$ are given respectively by

\begin{equation} \label{eq:kph}
    k_{P_H} = \frac{\left( \frac{(1+e^{ab})^{k_1-1}}{\left(1+e^{ab}\left(1-\zeta\right)\right)^{k_1}}-\frac{1}{1+e^{ab}}\right)^2}{\frac{(1+e^{ab})^{k_1-2}}{\left(1+e^{ab}\left(1-\zeta\right)\right)^{k_1}}-\frac{(1+e^{ab})^{2k_1-2}}{\left(1+e^{ab}\left(1-\zeta\right)\right)^{2k_1}}},
\end{equation}

\begin{equation} \label{eq:thetaph}
    \theta_{P_H} = \frac{\frac{(1+e^{ab})^{k_1-2}}{\left(1+e^{ab}\left(1-\zeta\right)\right)^{k_1}}-\frac{(1+e^{ab})^{2k_1-2}}{\left(1+e^{ab}\left(1-\zeta\right)\right)^{2k_1}}}{ \frac{(1+e^{ab})^{k_1-1}}{\left(1+e^{ab}\left(1-\zeta\right)\right)^{k_1}}-\frac{1}{1+e^{ab}}} .
\end{equation}

\end{lemma}
\begin{proof}
    The proof is available in Appendix \ref{app:proofPH}
\end{proof}

\begin{figure}[!htbp]
\centering
\includegraphics[trim={0  0  0mm 0mm},clip,width=.75\linewidth]{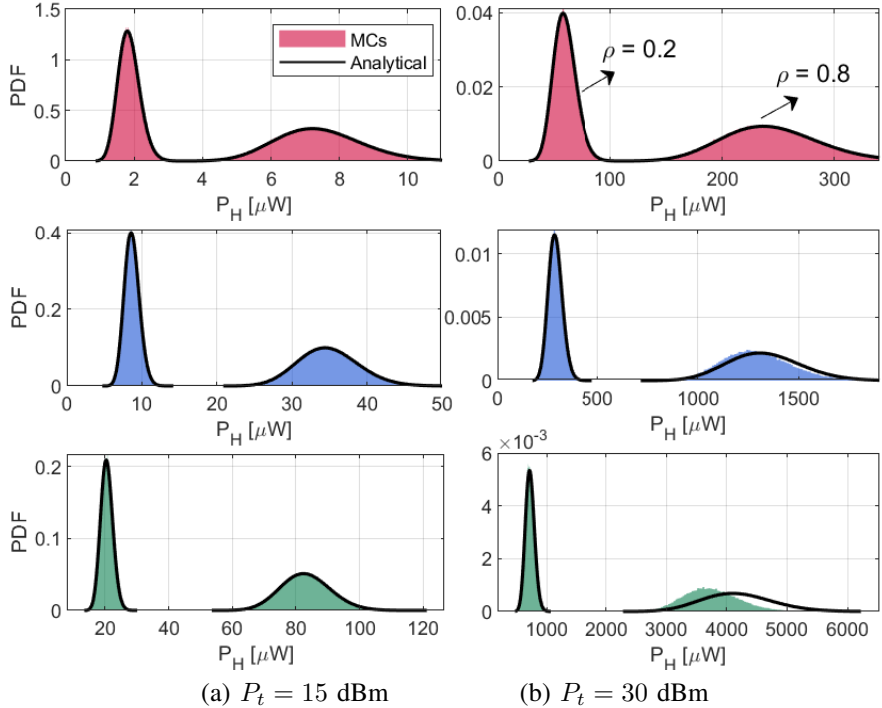}
\\
\small (a) $P_t=15$ dBm \hspace{15mm} (b) $P_t=30$ dBm
\caption{PDF of the harvested power $P_H$ for $N=30$ (red color), $N=65$ (blue color) and $N=100$ (green color), with $\rho=0.2$ (left slope) and $\rho=0.8$ (right slope). Another parameters are in Table \ref{t:param}.}
\end{figure}

\begin{remark}
The mean harvested power can be computed as 
\begin{IEEEeqnarray}{rCl} \label{eq:meanPH}
\bar{P}_H  &=& k_{P_H} \theta_{P_H} \nonumber
\\
&=& \frac{P_{th}}{1 - \frac{1}{1+e^{ab}}}\left(\frac{(1+e^{ab})^{k_1-1}}{\left(1+e^{ab}\left(1-\zeta\right)\right)^{k_1}}-\frac{1}{1+e^{ab}} \right)\hspace{-0.1cm}.
\end{IEEEeqnarray}
\end{remark}

\section{System Performance Analysis}\label{sec:analysis}
In order to explore the advantages in RIS-aided C-SWIPT-NOMA scenarios, in this section new expressions for the OP and SE of $D_1$ and $D_2$ in the RIS-aided  C-SWIPT-NOMA system considering generalized Nakagami-$m$ fading channels and non-linear energy harvesting model are presented, where we consider that target SINRs are determined by the devices’ QoS requirements. The OP metric is paramount to evaluate the reliability of the transmission in the 5G and 6G systems, specially in the URLLC use mode. 

\subsection{Device 1 Outage Probability}
Particularly, the outage behavior for $D_1$ occurs since $D_1$ cannot detect effectively $x_2$ and consequently its own message or when $x_2$ is detected successfully but an error occurs to decode $x_1$, $i.e.$, the outage occurs except for the case that both $D_2$ and itself message are decoded successfully. Mathematically it can be formulated by 
\begin{IEEEeqnarray}{rCl} \label{eq:PrD1}
    P_{\rm out}^{D_1} = 1 - {\rm Pr}({\rm SINR}_{D_1}^{x_2} \geq \gamma_{th2}, {\rm SINR}_{D_1}^{x_1} \geq \gamma_{th}),
\end{IEEEeqnarray}
where $\gamma_{th2} = 2^{R_2}-1$ and $\gamma_{th} = 2^{R_1}-1$, with $R_1$ and $R_2$ being the target rates of the $D_1$ and $D_2$, respectively. The following theorem provides the OP of $D_1$ for RIS-aided C-SWIPT-NOMA system.

\begin{theorem} \label{theo:OPD1}
Let $\xi_1=\frac{\gamma_{th}}{P_t(1-\rho)\alpha_1}$ and $\xi_2 = \frac{\gamma_{th2}}{P_t(1-\rho) (\alpha_2 - \alpha_1\gamma_{th2})}$, if $\alpha_2<\alpha_1\gamma_{th2}$ or $\frac{k_1\theta_1\beta_{ss}\beta_{s1}}{\xi_1}<\omega k_{P_H}\theta_{P_H}$, $P_{\rm out}^{D_1}=1$, otherwise  the closed-form expression for the OP of $D_1$ under Nakagami-$m$ fading is given by \eqref{eq:PoutD1rho0} when $\rho=0$ or $\omega=0$, or by \eqref{eq:PoutD1} at top of the next page when $\rho\neq0$ and $\omega \neq0$

\begin{numcases}{P_{\rm out}^{D_1}=} \label{eq:PoutD1rho0}
        \frac{\gamma\left(k_1, \frac{\xi_1 \sigma^2}{\theta_1\beta_{ss}\beta_{s1}} \right)}{\Gamma(k_1)}, & \qquad {\rm if } $\xi_1>\xi_2$ \nonumber
        \\
        \frac{\gamma\left(k_1,\frac{\xi_2 \sigma^2}{\theta_1\beta_{ss}\beta_{s1}}\right)}{\Gamma(k_1)}, & \qquad {\rm otherwise.}
\end{numcases}

\begin{figure*}[h]
\normalsize
\begin{numcases}{P_{\rm out}^{D_1} \approx} \label{eq:PoutD1}
   \frac{\gamma\left(k_1,\frac{\xi_{\ell} \sigma^2}{\theta_1\beta_{ss}\beta_{s1}}\right)}{\Gamma(k_1)} 
    +
    \frac{e^{\frac{\sigma^2}{\omega k_{P_H}\theta_{P_H}}}\xi_1 \omega k_{P_H} \theta_{P_H}}{\Gamma(k_1) } \frac{\Gamma\left(k_1,\xi_1\sigma^2 \left(\frac{1}{\theta_1\beta_{ss} \beta_{s1}} + \frac{1}{\xi_1 \omega k_{P_H} \theta_{P_H}}\right) \right)}{\left(\xi_1 \omega k_{P_H} \theta_{P_H}+\theta_1\beta_{ss} \beta_{s1} \right)^{k_1}} & $\xi_1 > \xi_2$ \nonumber
   \\
   \frac{\gamma\left(k_1,\frac{\xi_2 \sigma^2}{\theta_1\beta_{ss}\beta_{s1}}\right)}{\Gamma(k_1)} 
    +
    \frac{e^{\frac{\sigma^2}{\omega k_{P_H}\theta_{P_H}}}\xi_2 \omega k_{P_H} \theta_{P_H}}{\Gamma(k_1) } \frac{\Gamma\left(k_1,\xi_2\sigma^2 \left(\frac{1}{\theta_1\beta_{ss} \beta_{s1}} + \frac{1}{\xi_2\omega k_{P_H} \theta_{P_H}}\right) \right)}{\left(\xi_2 \omega k_{P_H} \theta_{P_H}+\theta_1\beta_{ss} \beta_{s1} \right)^{k_1}} & $\xi_1 < \xi_2$ 
\end{numcases}
\vspace*{4pt}
\hrulefill
\end{figure*}
\end{theorem}

\begin{proof}
    The proof is available in Appendix \ref{App:proofOPD1}.
\end{proof}

\subsection{Device 2 Outage Probability}
The outage behavior for the $D_2$ can occur  in two distinct ways: {\it 1)} $D_1$ detects effectively the $D_2$' message however the sum of the SINRs after MRC in $D_2$ is lower than the SINR of threshold {\bf or}  {\it 2)}  $D_1$ cannot detect effectively the $D_2$' message and the SINR becoming of the $S$ $\rightarrow$ RIS $\rightarrow$ $D_2$ is lower than the SINR of threshold. Therefore, the OP of $D_2$ can be formulated as 
\begin{IEEEeqnarray}{rCl} \label{eq:PrD2}
    P_{\rm out}^{D_2} &=&  
    {\rm Pr}({\rm SINR}_{D_1}^{x_2} < \gamma_{th2},{\rm SINR}_{D_2,S}^{x_2} < \gamma_{th2}) \nonumber
    \\
    &+& 
    {\rm Pr}({\rm SINR}_{D_1}^{x_2} \geq \gamma_{th2}, {\rm SINR}_{D_2}^{x_2} < \gamma_{th2}).
\end{IEEEeqnarray}

\begin{theorem} \label{theo:PoutD2}
  Let us define $I = 1 - e^{\frac{\sigma^2\gamma_{th2}}{k_{P_H}\theta_{P_H}\beta_{12}}}-e^{\frac{-\sigma^2\gamma_{th2}}{P_t\beta_{ss}\beta_{s2}N(\alpha_2-\alpha_1\gamma_{th2})}}   \frac{\left( e^{\frac{\sigma^2\gamma_{th2}}{P_t\beta_{ss}\beta_{s2}N(\alpha_2-\alpha_1\gamma_{th2})} - \frac{\gamma_{th2 \sigma^2}}{\beta_{12k_{P_H}\theta_{P_H}}}} - 1 \right)}{\frac{\beta_{12}k_{P_H}\theta_{P_H}}{P_t\beta_{ss}\beta_{s2}N(\alpha_2-\alpha_1\gamma_{th2})}-1}$, and let the following variable given as $\chi =\frac{\gamma\left(k_1,\frac{\xi_2 \sigma^2}{\theta_1\beta_{ss}\beta_{s1}}\right)}{\Gamma(k_1)} + \frac{e^{\frac{\sigma^2}{\omega k_{P_H}\theta_{P_H}}}\xi_2 \omega k_{P_H} \theta_{P_H}}{\Gamma(k_1) } \frac{\Gamma\left(k_1,\xi_2\sigma^2 \left(\frac{1}{\theta_1\beta_{ss} \beta_{s1}} + \frac{1}{\xi_2\omega k_{P_H} \theta_{P_H}}\right) \right)}{\left(\xi_2 \omega k_{P_H} \theta_{P_H}+\theta_1\beta_{ss} \beta_{s1} \right)^{k_1}}$, for the case when $\alpha_2<\alpha_1\gamma_{th2}$, then $P_{\rm out}^{\textsc{u}_2}=1$, otherwise, then the closed-form expression for the OP of $D_2$ under Nakagami-$m$ fading is given by 
\end{theorem}

\begin{numcases}{P_{\rm out}^{\textsc{u}_2} \approx } \label{eq:PoutD2}
\hspace{-0.1cm} 1 - e^{\frac{-\sigma^2 \gamma_{th2} }{P_t\beta_{ss}\beta_{s2}N(\alpha_2-\alpha_1\gamma_{th2})}}, \hspace{-0.35cm} & $\text{if}\; \rho=0$ \nonumber 
\\
\hspace{-0.1cm}I + \chi \left(\hspace{-0.1cm}1 - I- e^{\frac{-\sigma^2 \gamma_{th2} }{P_t\beta_{ss}\beta_{s2}N(\alpha_2-\alpha_1\gamma_{th2})}}\right) \hspace{-0.35cm} & $\text{if} \; \rho \neq 0$, 
\end{numcases}

\begin{proof}
    Please refer to Appendix \ref{App:proofOPD2}.
\end{proof}

\subsection{Device 1 Ergodic Rate}
In this subsection, to understand better the RIS-aided C-SWIPT-NOMA system with a non-linear EH model, we proposed an upper bound for the ER of $D_1$.

\begin{theorem}
Assuming the $D_1$ can successfully detect $D_2$ and itself message, then the upper bound spectral efficiency for $D_1$ under Nakagami-$m$ fading can be calculated as
\begin{figure*}[h]
\normalsize
\begin{equation} \label{eq:ERD1}
    R_1 \leq \log_2\left(1 - \frac{\beta_{ss}\beta_{s1}k_1\theta_1  (1-\rho) P_t \alpha_1 e^{-\frac{\sigma^2}{\omega k_{P_H}\theta_{P_H}}}{\rm Ei}\left(-\frac{\sigma^2}{\omega k_{P_H}\theta_{P_H}}\right)}{ \omega  k_{P_H} \theta_{P_H} }\right)
\end{equation}
\vspace*{4pt}
\hrule
\end{figure*}

\end{theorem}

\begin{proof}
The ER of $D_1$ can be expressed through $\mathbb{E} \left[R_1\right] = \mathbb{E}\left[ \log_2(1+{\rm SINR}_{D_1}^{x_1})\right]$. Therefore, by utilizing the Jensen' inequality, we can obtain a upper bound for rate of $D_1$ as $\mathbb{E}[R_1] \leq \log_2\left(1 + \mathbb{E}\left[ \frac{|h_1|^2  (1-\rho) P_t \alpha_1 }{ |h_{SI}|^2 P_H + \sigma^2}\right] \right) $, where 

\begin{IEEEeqnarray}{rCl}
    \mathbb{E}\left[\frac{1}{|h_{SI}|^2 P_H+\sigma^2}  \right] &=& \frac{1}{\omega k_{P_H}\theta_{P_H}}\int_{0}^{\infty} \frac{e^{-\frac{x}{\omega k_{P_H}\theta_{P_H}}}}{x+\sigma^2} dx  \nonumber \\ 
    &=&  \frac{-e^{\frac{\sigma^2}{\omega k_{P_H}\theta_{P_H}}} {\rm Ei}\left(-\frac{\sigma^2}{\omega k_{P_H}\theta_{P_H}}\right)}{\omega k_{P_H}\theta_{P_H}}, \nonumber \\
\end{IEEEeqnarray}
where \cite[3.352.4]{gradshteyn2014table} is utilized.
\end{proof}

\subsection{Device 2 Ergodic Rate}
An upper bound for the ER of $D_2$ is proposed in the following theorem.
\begin{theorem} \label{t:R2}
Assuming the $D_2$ decoded successfully the message of $x_2$ from $S$ $\rightarrow${\rm RIS}$\rightarrow$ $D_2$ as well as the rebuilt message $\hat{x}_2$ from the cooperative link (since $D_1$ decoded successfully $x_2$ in order to relay it), the upper bound for ER of $D_2$ is given as

\begin{figure*}[h]
\normalsize
\begin{IEEEeqnarray}{rCl} \label{eq:ERD2}
    R_2 &\leq& \log_2\left(1 + \frac{\alpha_2\sigma^2e^{\frac{\sigma^2}{N\beta_{ss}\beta_{s2}P_t \alpha_1}}}{N\beta_{ss}\beta_{s2}P_t \alpha_1^2}  \Gamma\left(-1,\frac{\sigma^2}{N\beta_{ss}\beta_{s2}P_t \alpha_1}\right) +  \frac{k_{P_H}\theta_{P_H}\beta_{12}}{\sigma^2}\right) 
\end{IEEEeqnarray}
\hrule
\vspace*{4pt}
\end{figure*}
\end{theorem}

\begin{proof}
Realizing the same step as done in Appendix, we obtain Eq. \eqref{eq:ERD2}, where 

\begin{IEEEeqnarray}{rCl}
    \mathbb{E}\left[ \frac{|h_2|^2 P_t \alpha_2 }{|h_2|^2 P_t \alpha_1 + \sigma^2} \right] &=& \int_{0}^{\infty} \frac{|h_2|^2 P_t \alpha_2}{ |h_2|^2 P_t \alpha_1+\sigma^2} f_{|h_2|^2} d|h_2|^2 \nonumber
    \\
    &\overset{(i)}{=}& \frac{ \alpha_2}{\alpha_1N} \int_{0}^{\infty} \frac{ x  e^{-\frac{x}{N}} }{  x + \frac{\sigma^2}{\beta_{ss}\beta_{s2}P_t \alpha_1} } dx \nonumber
    \\
    &=&  \frac{\alpha_2\sigma^2}{N\beta_{ss}\beta_{s2}P_t \alpha_1^2} e^{\frac{\sigma^2}{N\beta_{ss}\beta_{s2}P_t \alpha_1}} \nonumber \\
    &\times& \Gamma\left(-1,\frac{\sigma^2}{N\beta_{ss}\beta_{s2}P_t \alpha_1}\right),
\end{IEEEeqnarray}
utilizing \eqref{eq:normpdf} and \cite[3.383.10]{gradshteyn2014table} is utilized in (i).
\end{proof}

\begin{corollary}
When $P_t\rightarrow\infty$, the ER of $D_2$, $R_2^{\infty}$ can be computed as
\end{corollary}

\begin{numcases}{=} \hspace{-0.1cm}
\log_2\left(1+\frac{\alpha_2}{\alpha_1}\right),& \hspace{-0.3cm} $\rho=0$ \nonumber
\\
\hspace{-0.1cm} \log_2\left(\hspace{-0.09cm}1\hspace{-0.09cm}+\hspace{-0.1cm}\frac{\alpha_2}{\alpha_1}\right) \hspace{-0.12cm} - \hspace{-0.1cm} \frac{e^{\frac{\sigma^2\left(1+\frac{\alpha_2}{\alpha_1}\right)}{P_{th}\beta_{12}}}{\rm Ei}\left(- \frac{\sigma^2\left(1+\frac{\alpha_2}{\alpha_1}\right)}{P_{th}\beta_{12}}\right)}{\ln(2)}, & \hspace{-0.3cm}  $\rho\neq0$ \nonumber\\ \hspace{-2cm} 
\label{eq:R2assym}
\end{numcases}

\begin{proof}
When $\rho=0$, it reaches the conventional NOMA case and can be solved trivially, otherwise, when $\rho\neq0$, and $P_t\rightarrow \infty$, due to the non-linear EH model, $P_{H} = P_{th}$, therefore the ER of $D_2$ is given as 
\begin{equation}
    R_2^{\infty} = \mathbb{E}\left[ \log_2\left(1 + \frac{\alpha_2}{\alpha_1} + \frac{P_{th}|h_{12}|^2}{\sigma^2} \right) \right],
\end{equation}
hence
\begin{IEEEeqnarray}{rCl}
    R_2^{\infty} &=& \int_{0}^{\infty} \log_2\left( 1 + \frac{\alpha_2}{\alpha_1} + \frac{P_{th} \beta_{12} x}{\sigma^2}\right) e^{-x} dx \nonumber 
    \\
    &\overset{(i)}{=}& \frac{\sigma^2}{P_{th}\beta_{12} \ln(2)} \int_{0}^{\infty} \ln\left(1 + \frac{\alpha_2}{\alpha_1} + x\right) e^{-\frac{\sigma^2x}{P_{th}\beta_{12}}} \nonumber
    \\
    &=&  \log_2\left(1+\frac{\alpha_2}{\alpha_1}\right) - \frac{e^{\frac{\sigma^2\left(1+\frac{\alpha_2}{\alpha_1}\right)}{P_{th}\beta_{12}}}}{\ln(2)} {\rm Ei}\left(- \frac{\sigma^2\left(1+\frac{\alpha_2}{\alpha_1}\right)}{P_{th}\beta_{12}}\right), \nonumber \\
\end{IEEEeqnarray}
where in $(i)$ \cite[4.337.1]{abramowitz1972handbook} is utilized.
\end{proof}

\section{Simulation Results}\label{sec:simul}
In this section, we aim to confirm through Monte-Carlo simulations (MCs) the accuracy of our previous mathematical analysis and illustrate the achievable enhanced performance of the RIS-aided cooperative FD-SWIPT-NOMA system. The simulations results are averaged over $10^6$ realizations. Unless stated otherwise, the parameter values adopted in this section are presented in Table \ref{t:param}. In the following numerical results, the Monte-Carlo simulations curves are labeled as "MCs", and the derived analytical expression-based curves are labeled as "Analytical". We use dashed lines to represent the $D_1$ performance, while $D_2$ is represented by solid lines. In addition, color red, blue and green, denotes the simulations for $N=30, 65$, and $100$ RIS elements, respectively.

\begin{table}[htbp!] 
\caption{Adopted Simulation Parameters.}
\footnotesize
\label{t:param} 
\begin{center}
\begin{tabular}{ll}
\toprule
\bf Parameter  & \bf  Value\\
\toprule
\multicolumn{2}{c}{\bf RIS-aided Cooperative FD-SWIPT-NOMA system}\\
\toprule
Transmit power          &   $P_t=[0,50]$ [dBm]\\
Noise power             &  $\sigma^2=-96$  dBm/Hz\\
Bandwidth               & $B=1$ MHz\\
\# RIS elements         &   $N = \{30;65;100\}$\\
Power Allocation coefficients            &   $\alpha_1$ = 0.9, $\alpha_2$ = 0.1\\ 
Target Rate &   $R_1 = 1.5$ $R_2=0.5$ [bits/s]\\
Residual Self-Interference              &   $h_{SI}\sim \mathcal{CN}(0,\, \omega)$, \\
& with  $\omega=$\,\footnotesize{$-[\infty;30;15]$} [dB]\\
\toprule
\multicolumn{2}{c}{\bf Non-Linear Energy Harvesting  Parameters}\\
\toprule
EH coefficient   & $\rho \in [0,1]$\\
EH model constants  & $a=150$; $b=0.014$ \cite{9736445,8652416}\\
Max. RF (harvested)  power & $P_{th} = 24$ [mW] \cite{9736445,8652416}\\
\toprule
\multicolumn{2}{c}{ \bf Channel Parameters}\\
\toprule
Channel Model ({\footnotesize $S$-RIS/RIS-$D$s})  &  Nakagami-$m$\\
Shape Parameter $S$-RIS                   & $m_{ss}$ = 3.5 \\
Shape Parameter RIS-$D$s                & $m_{s1}=2$ $m_{s2}=1$\\
Spread Parameter                         & $\Omega_{s1} = \Omega_{s2} = 1$\\
Channel Model (D2D)                      & $h_{12}\sim$ $\mathcal{CN}$(0,$\beta_{12}$) \\
Path losses &                        $\beta_{ss}=-30$dB, $\beta_{s1}=-30$dB \\
& $\beta_{s2}=-40$dB, $\beta_{12}=-15$dB\\
\toprule
 \end{tabular}
 \end{center}
\end{table}

\subsection{Ergodic Rate {\it vs} Transmit Power}

Fig. \ref{fig:ERvsPT} depicts the ER of $D_1$ and $D_2$  {\it vs} $P_t$ for three values of $N=\{30;65;100\}$ RIS-elements. Firstly we can notice that the derived upper bound expressions given by Eq. \eqref{eq:ERD1} and Eq. \eqref{eq:ERD2} respectively are very tight when $\rho=0$; furthermore, when $\rho=0.2$, the upper bound is especially tight fow low values of $P_t$, for any $N$, as $P_t$ increases, the ratio between the ER and the derived upper bound decreases. We also can note that after a specific value of $P_t$, the derived equations become inaccurate. This can be justified by the adoption made in Appendix \ref{app:proofPH}, where $aP_t\rho \beta_{ss}\beta_{s1}\theta_1 \ll 1$. We separate the accurate region (I) (denoted by low power regime and green color) and the inaccurate region (II) (high power regime red color). Notice that as stated, the regions vary according to $N$, $\rho$, and $P_t$ (since $\beta_{ss}$, $\beta_{s1}$ and $a$ are fixed values), when $N=30$, the region (II) occurs for $P_t>49$ dBm, as well as $P_t>43$ and $P_t>39$ for $N=65$ and $N=100$ respectively.

\begin{figure} 
\centering
\includegraphics[trim={0mm 0mm 0mm 0mm}, clip,width=.75\linewidth]{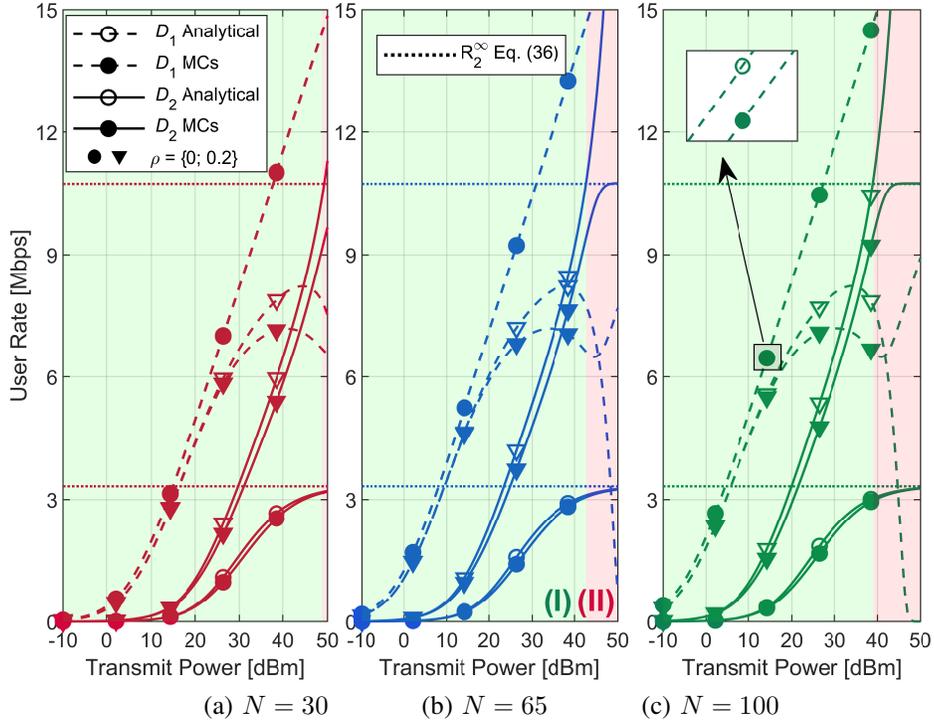}
\\
\small \hspace{5mm} (a) $N=30$ \hspace{10mm} (b) $N=65$ \hspace{10mm} (c) $N=100$
\caption{ER {\it vs} $P_t$ for three values of number of elements of RIS $N=\{30;65;100\}$ in RIS-aided C-SWIPT-NOMA system, with $\rho=\{0;0.2\}$.}
\label{fig:ERvsPT}
\end{figure}

Concerning the RIS-aided C-SWIPT-NOMA system performance, we can see  that when $\rho=0$, this configuration represents the conventional NOMA system, thus, for any value of $N$, $D_1$ naturally reaches higher rates than $D_2$, furthermore, $D_2$ reaches a maximum rate when $P_t\rightarrow \infty$, given by Eq. \eqref{eq:R2assym}. For the cooperative scenario, $\rho=0.2$, we can see that rate of $D_1$ increases initially till reaching a maximum value then it decreases and afterward increases again. This behavior can be explained by the non-linear EH model, once $P_t$ increases, $D_1$ can harvest a higher amount of power, and when it reaches the saturation $P_{th}$, the interference in the denominator of Eq. \eqref{eq:SINRd1x1} is limited, hence by increasing the transmit power, the rate of $D_1$ increases again. In addition, we can confirm that the value of maximum rate reached by $D_2$ is not dependent of $N$, $i.e.$, the RIS cannot contribute to increase the rate of $D_2$ when $P_t \rightarrow \infty$, however it can be achieved with lower power when $N$ increases. It is confirmed in Eq. \eqref{eq:R2assym}, where we can see that $R_2^{\infty}$ can vary in function of $\alpha_1,\alpha_2,P_{th},\beta_{12}$ and $\sigma^2$. 
\subsection{Harvested Power {\it vs} Transmit Power}

Fig. \ref{fig:PHvsPT} illustrate the behaviour of the harvested power $P_H$ parameterized in the transmitted power $P_t$ for the MCs and the analytical derived result. We firstly notice that Eq. \eqref{eq:meanPH} is very accurate for low and middle values of $P_{t}$, {\it i.e.}, $P_t\leq 30$ dBm. Besides, in this region we can see that the harvested power increases linearly with the transmitted power. We also can see that for higher values of $P_t$, the $P_H$ saturates in $P_{th}$, independently of value of $N$. This is expected since a non-linear EH model is adopted.

\begin{figure}[!htbp]  
\centering
\includegraphics[trim={0mm 0mm 0mm 0mm}, clip,width=.75\linewidth]{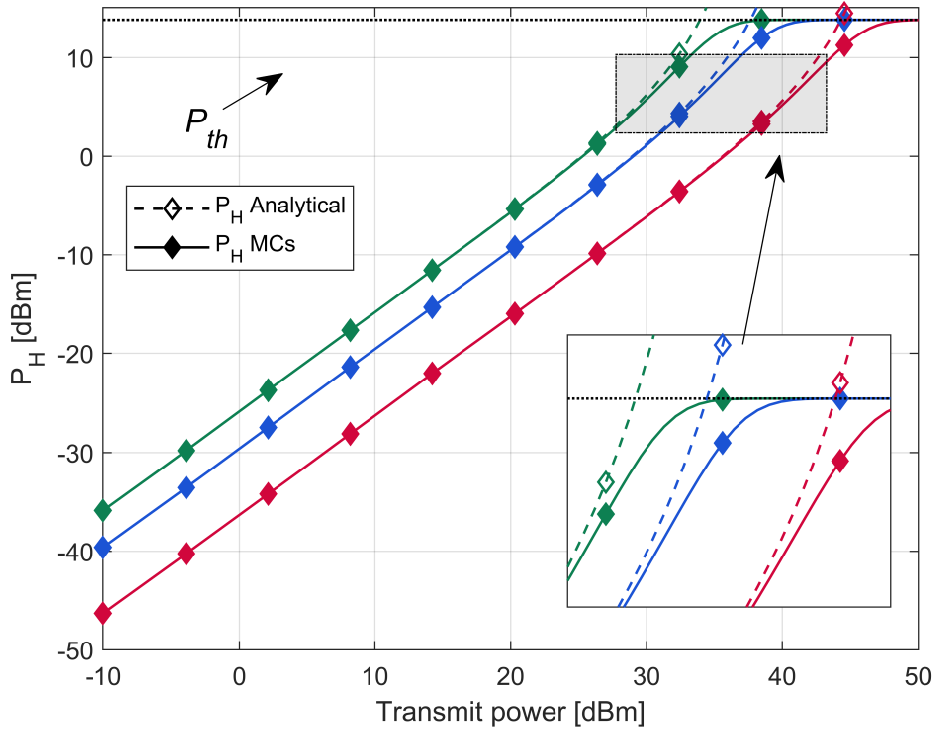}
\caption{$P_H$ {\it vs} $P_t$ with non-linear EH model, where $\rho=0.8$. Here we adopt $N=\{30;65;100\}$ as red, blue and green color respectively.  }
\label{fig:PHvsPT}
\end{figure}

\subsection{Outage Probability {\it vs} Transmit Power}

Fig. \ref{fig:OPvsPT} depicts the OP performance {\it vs} the transmit power in the node $S$. Firstly, one can infer that the derived equations \eqref{eq:PoutD1} and \eqref{eq:PoutD2} are very accurate for the values of $\rho$ and $P_t$ from $-10$ to $50$ dBm till the order of $10^{-5}$ for the OP. 
Besides, one can infer that in the cooperative scenario the reliability of device $D_2$ increases considerably, mainly when $N$ assumes high values; in contrast, due to power drainage to operate as a relay, the device $D_1$ has a loss of performance, as also confirmed in Fig. \ref{fig:ERvsPT} and \ref{fig:ERvsN}, considering $\rho=0.2$ and $0.5$, respectively. Herein, to understand the potential of the studied RIS-aided cooperative FD-SWIPT-NOMA system, we consider the ideal case where $\omega=0$, $i.e.$ there is no residual self-interference cancellation in the SIC stage. We can see that although $D_1$ operates in worst conditions when $\rho=0.2$, its OP performance degradation is marginal when related to the OP gain obtained by $D_2$, indicating that the cooperative scenario can be very interesting for $D_1$ if there is no residual self-interference and mainly for $D_2$ due to the remarkable improvement in the OP performance gain attained. Moreover, one can see that by increasing the value of $N$, the OP performance loss in $D_1$ can not be mitigated, while the gain of performance of $D_2$ becomes lower.

\begin{figure}[htbp!]
\normalsize
\centering
\includegraphics[trim={0mm 0mm 0mm 0mm}, clip,width=.75\linewidth]{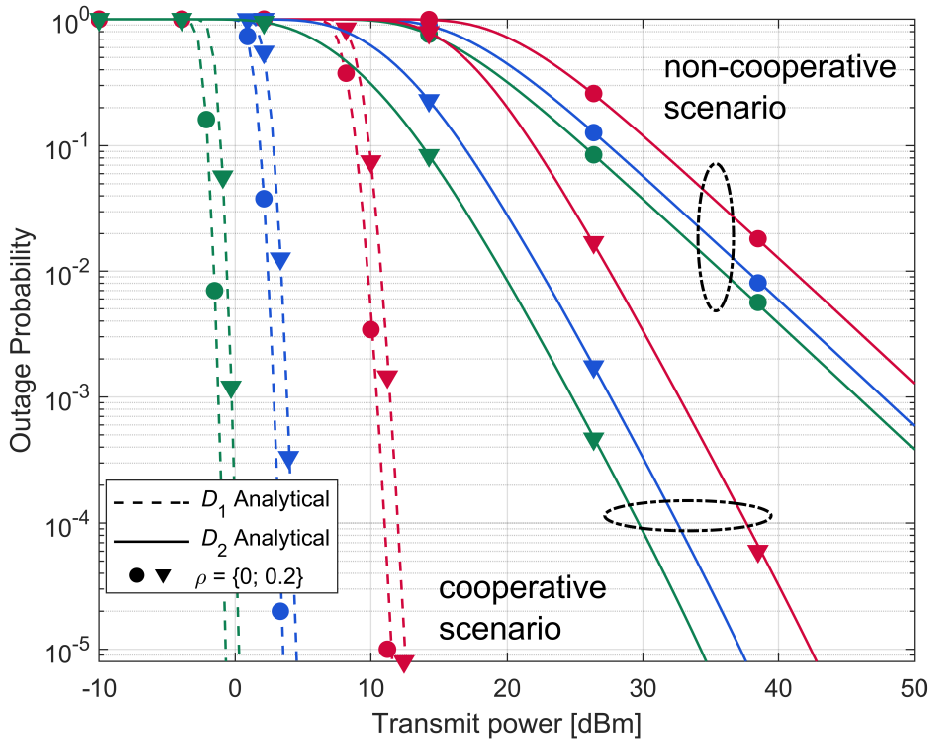}
\caption{OP {\it vs} $P_t$ for $N=\{30,65,100\}$ (red, blue and green, respectively), and $\rho=\{0;0.2\}$, where we consider the ideal case $\omega=0$, in order to see how is the best performance of the system with the considered parameters.}
\label{fig:OPvsPT}
\end{figure}

\subsection{Ergodic Rate {\it vs} RIS elements}

Fig. \ref{fig:ERvsN} illustrates the impact of the number of elements of RIS on the device rate ($D_1$ and $D_2$) when $P_{t}=15$ dBm for $\rho=0$ and $\rho=0.5$. 
Notice that our derived upper bounds given by Eqs. \eqref{eq:ERD1} and \eqref{eq:ERD2} can represent the behavior of the rate of both devices. Also, one can see that for the non-cooperative scenario ($\rho=0$), naturally, $D_1$ can achieve a substantial performance gain over $D_2$; furthermore, when $N$ increases, the enhanced performance is further highlighted in $D_1$ than $D_2$, it is due to the RIS being configured to boost $D_1$ (coherent phase shift matrix for $D_1$), while the rate of $D_2$ increases marginally (from $\approx 0.15$ to $\approx 0.4$) when $N$ becomes higher (once that random phase shift matrix is set). On the other hand, in the cooperative scenarios ($\rho=0.5$), one can notice that the $D_1$ rates are lower than in the former case, which is expected since $D_1$ drains energy to operate as a relay. Besides, notably, the impact of the $\omega$ in the ER of $D_1$ is very harmful if we assume high values of $\omega$ and its impact is highlighted for large values of $N$. However, in such a scenario, $D_2$ can achieve higher rates, and now, it effectively increases as $N$ increases, $e.g.$, the rate of $D_2$ increases from $\approx 0.7$ bits/s to $2.5$ bits/s, obtaining a very interesting gain in terms of ER that was not possible in the non-cooperative case. 
Effectively, RIS-aided systems can contribute to the cooperative system in order to provide better conditions for the cell-edge devices and consequently provide better fairness indexes for RIS-aided cooperative FD-SWIPT-NOMA communication systems.

\begin{figure}[htbp!]
\normalsize
\centering
\includegraphics[trim={0mm 0mm 0mm 0mm}, clip,width=.75\linewidth]{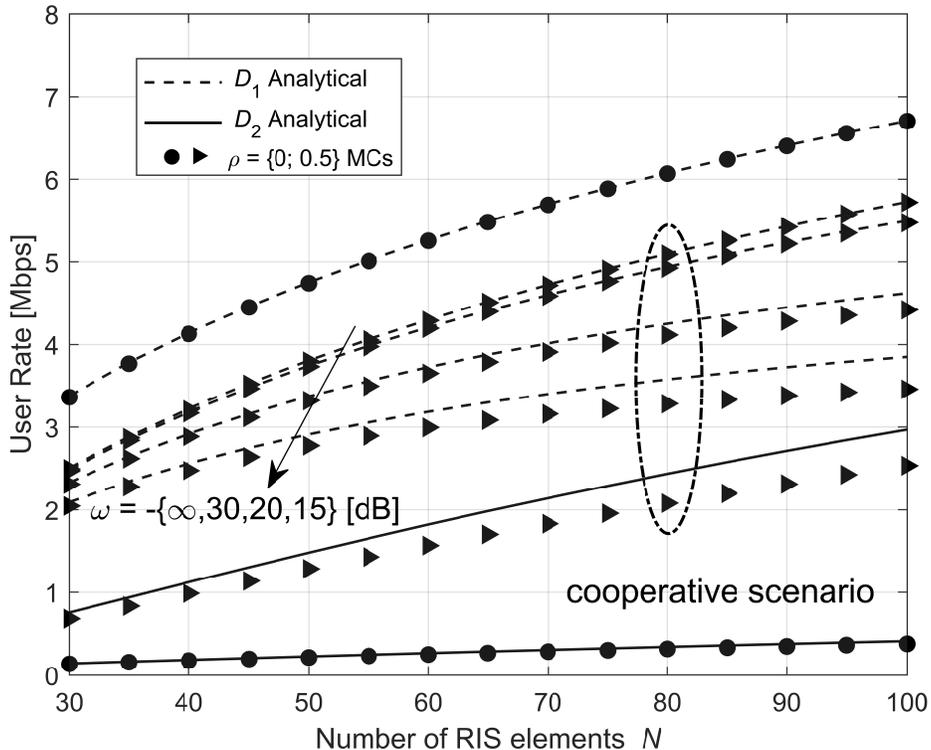}
\caption{ER {\it vs} $N$ for cooperative and non-cooperative scenario $\rho=\{0;0.5\}$, $P_t=15$ dBm and $\omega=\{-45,-30,-15\}$ dB. }
\label{fig:ERvsN}
\end{figure}

Next, to understand the impact of the residual self-interference on the OP, in Fig. \ref{fig:OPvsPT_omega} we studied how the performance of $D_1$ can be affected when the SIC operation is degraded by increasing the residual SI factor $\omega$; herein we have evaluated the OP degradation for $\omega\in\{-\infty,\, -20, \, -15\}$ dB and considering $N=100$ RIS antenna-elements. As a benchmark, we consider the linear EH model with an efficiency of $\eta=0.8$ to understand its impacts and differences when adopting a more realistic non-linear EH model. One can see that naturally, $D_1$ achieves better performance when $\omega=0$; hence, as $\omega$ increases, its performance is severally degraded for both linear and non-linear EH models. Furthermore, it is notable that the linear EH model has an additional negative impact on the OP performance of $D_1$, for any value of $\omega$. The reason is that the linear EH model can harvest much more power than the non-linear model, leading to higher interference. Finally, we also can see that the linear and non-linear EH model reveal different asymptotic operational behaviors, {\it i.e.}, when $P_t \rightarrow \infty$, the linear EH model saturates in an OP floor where this behavior already has been reported in the literature \cite{8026173,tran2021swipt}. In contrast, one can see that the asymptotic behavior of the non-linear EH model reveals a continuous increase in the OP of $D_1$ till to reach a maximum (when $P_H$ saturates), and then it is expected to start decreasing again according to the behavior illustrated in Fig. \ref{fig:ERvsPT}. It is paramount to understand that the main reason to this behavior is that the linear model provides false considerations about the EH operation, {\it i.e.}, when $P_t \rightarrow \infty$, $P_H \rightarrow \infty$, which is not physically possible, in contrast, for the non-linear EH model, when $P_t \rightarrow \infty$, $P_H\rightarrow P_{th}$; hence, for higher values of power, it is expected that $P_{\rm out}^{D_1}$ become 0.

\begin{figure}[h!]
\normalsize
\centering
\includegraphics[trim={0mm 0mm 0mm 0mm}, clip,width=.75\linewidth]{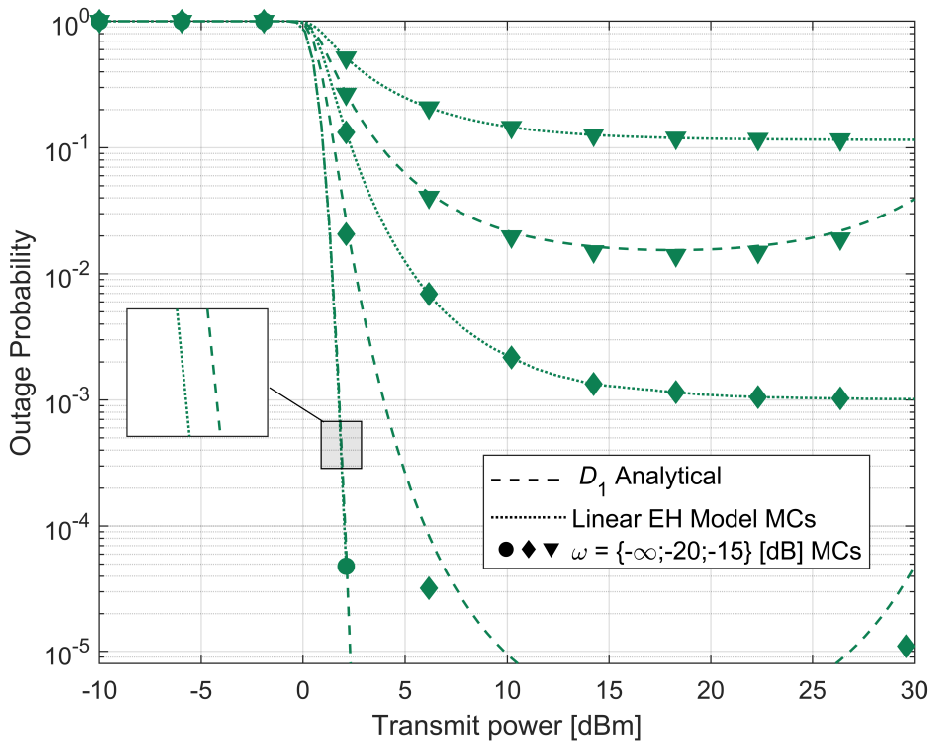}
\caption{OP {\it vs} $P_t$ with $N=100$ and $\rho=0.5$, for non-linear EH model and linear EH model.}
\label{fig:OPvsPT_omega}
\end{figure}
\vspace{-0.5cm}
\subsection{Outage Probability {\it vs} Energy Harvest Coefficient}

To gain further insights, in Fig. \ref{fig:OPvsPT_NandRHO} we plot the OP of both devices against the EH coefficient in the range $\rho\in[0.0;\, 1.0]$, considering two different transmission power $P_t=15$ and $P_t=30$ [dBm], with $\omega=\{-5,-10,-15\}$ dB and $N=100$ RIS elements. Firstly, one can notice the significant reliability improvement on $D_2$ by changing the $\rho$ value, from low to intermediate values, considering both values of $P_t$. Besides, one can see how it is important to reduce the residual self-interference, in order to $D_1$ to be able to cooperate. We notice that when the system is not well designed to remove the self-interference, the $D_2$ performance can become worst, $e.g.$, when SI factor is very high, {\it i.e.,} $\omega=-5$ dB, $D_2$ performance becomes worst if $\rho>0.65$, $\rho>0.45$ for $P_t=15$ and $P_t=30$ dBm respectively. It is justified due to the interference increasing over $D_1$, which impacts the decoding of $D_2$' message, and thus, $D_1$ can not operate as a relay. Also, notice that when $\omega$ assumes lower values, $D_1$ can drain and cooperate with high values of $\rho$, while its own OP performance is just marginally impacted.

\begin{figure}[htbp!]
\normalsize
\centering
\includegraphics[trim={0mm 0mm 0mm 0mm}, clip,width=.75\linewidth]{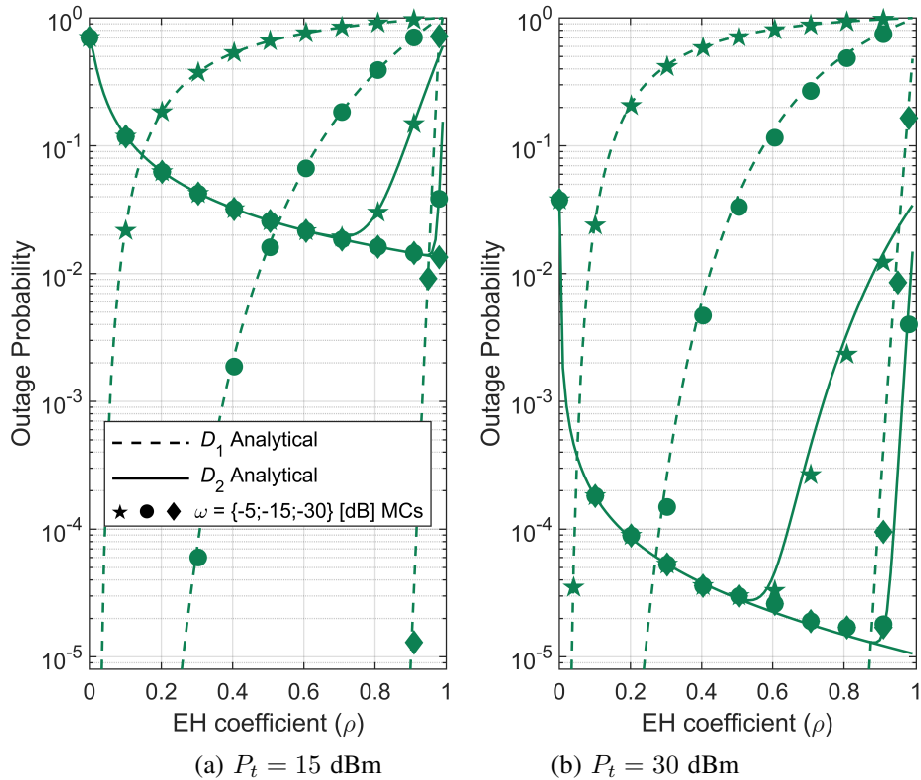}\\
\small (a) $P_t=15$ dBm \hspace{20mm} (b) $P_t=30$ dBm
\caption{OP {\it vs} $\rho$ for $N=100$ under (a) $P_t=15$ [dBm];  (b) $P_t=30$ [dBm]. Impact of SI factor $\omega=\{-5,-15,-30\}$ [dB] on the OP of $D_1$ and $D_2$.}
\label{fig:OPvsPT_NandRHO}
\end{figure}

\vspace{-0.5cm}
\subsection{Outage Probability {\it vs} Threshold Power} 

Fig. \ref{fig:OPvsPth} shows the impact of $P_{th}$ on the OP performance of both devices. When the $D_1$ has a higher capacity to harvest power (higher $P_{th}$ values),  the performance of $D_2$ is further improved. On the other hand, the self-interference also is increased and consequently worst its own performance. Also, the negative impact of the SI on the performance is notable, confirming
how it is paramount to reduce $\omega$ in order to does not jeopardize the $D_1$ performance, mainly when $P_{th}$ assume high values.

\begin{figure}[h]
\normalsize
\centering
\includegraphics[trim={0mm 0mm 0mm 0mm}, clip,width=.75\linewidth]{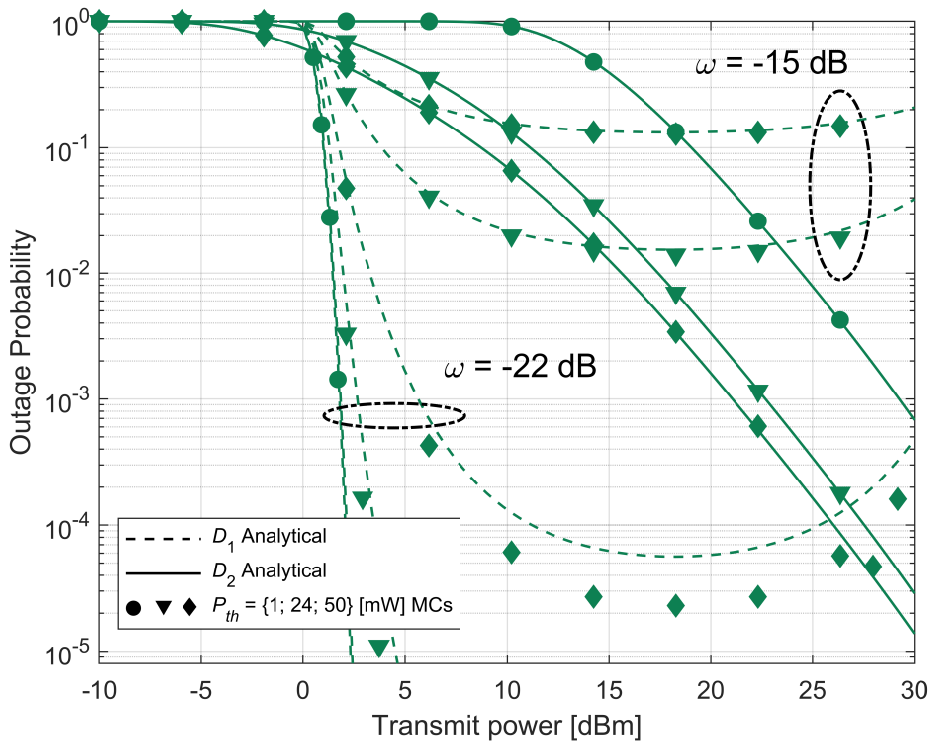}
\caption{OP {\it vs} $P_t$ for $N=100$, $\rho=0.5$ and $\omega=\{-22;-15\}$ [dB]. Herein,  $P_{th}=\{1;25;50\}$ mW.}
\label{fig:OPvsPth}
\end{figure}

\section{Conclusions}\label{sec:concl}
This work has investigated the downlink communication reliability and throughput of both the cell-center and the cell-edge device in RIS-aided cooperative SWIPT-NOMA system operating under Nakagami-$m$ fading; the cell-center device leverages of SWIPT technique equipped with a non-linear EH model to act as a FD-DF relay for assisting the cell-edge device. 
It was shown that under certain conditions, the harvested power in the cell-center device can be approximated as a Gamma random variable, allowing flexible modeling of the system. We have derived tractable closed-form expressions for the performance of both types of users, including: \textbf{\textit{a}}) upper bound for ER of a cell-center/edge devices; \textbf{\textit{b}}) tight expression for OP of both devices; and \textbf{\textit{c}}) maximum rate that the cell-edge user can attain in the cooperative mode. In addition, we assessed the impact of the RIS dimension on both metrics adopted. Also, it is demonstrated that the RIS can effectively improve the data rate of the cell-edge device when the cell-center device is cooperating. Besides, we show how paramount is to mitigate the residual self-inference till when $N$ assumes large values. Simulation results validate the correctness and the effectiveness of the developed theoretical analysis while demonstrate the advantages of the analyzed RIS-aided cooperative SWIPT-NOMA system over the non-cooperative RIS-aided system.


\null
\appendices 
\null 
\section{Proof of Lemma \ref{lemma:PH}} \label{app:proofPH}

Let us define $Y=\frac{1}{1+e^{-a(\rho P_{in}-b)}}$, thus we can write the first moment of $Y$ as

\begin{equation} \label{eq:appB1}
\mathbb{E}[Y] = \mathbb{E}\left[\frac{1}{1+e^{-a(\rho P_{in}-b)}}\right] = \hspace{-0.1cm} \int_{0}^{\infty} \hspace{-0.2cm} \frac{f_{|h_1|^2}(x)}{1+e^{-a(\rho P_t \beta_{ss} \beta_{s1} x -b)}} dx.
\end{equation}

To the best of the authors' knowledge, the integral in Eq. \eqref{eq:appB1} does not admit closed-form expression, here we will utilize the following approximation $e^{-a P_t \rho \beta_{ss}\beta_{s1}x}\approx 1-aP_t\rho\beta_{ss}\beta_{s1}x$. This approximation is reasonable for lower values of transmit power ($P_t$) and RIS elements ($N$) since  $a\rho\beta_{ss}\beta_{s1} \ll 1$, once the product of path losses naturally assume low values, thus substituting \eqref{eq:normpdf} in \eqref{eq:appB1}, we obtain

\begin{IEEEeqnarray}{rCl} \label{eq:appB2}
\mathbb{E}[Y] &\approx& \frac{1}{\Gamma(k_1) \theta_1^{k_1}(1+e^{ab}) } \hspace{-0.05cm} \int_{0}^{\infty} \hspace{-0.35cm} x^{k_1-1} e^{-x\left(\frac{1}{\theta_1}-\frac{e^{ab}aP_t\rho\beta_{ss}\beta_{s1}}{1+e^{ab}}\right)}    dx \nonumber
\\
&\approx& 
\frac{(1+e^{ab})^{k_1-1}}{\left(1+e^{ab}\left(1-aP_t\rho\beta_{ss}\beta_{s1}\theta_1\right)\right)^{k_1}}, 
\end{IEEEeqnarray}
where \cite[3.381.4]{gradshteyn2014table} is utilized. Realizing same step above, similarly, the second moment of $Y$ can be given by 
\begin{equation} \label{eq:appB3}
    \mathbb{E}[Y^2] \approx \frac{(1+e^{ab})^{k_1-2}}{\left(1+e^{ab}\left(1-2aP_t\rho\beta_{ss}\beta_{s1}\theta_1\right)\right)^{k_1}}. 
\end{equation}

Therefore utilizing the M.M. technique, the shape and scale parameter of $P_H$ in Eq. \eqref{eq:EHmodel} can be evaluated as follows
\begin{equation} \label{eq:appB4}
    k_{P_H} = \frac{\left( \mathbb{E}[Y]-\frac{1}{1+e^{ab}}\right)^2}{\mathbb{E}[Y^2]-\mathbb{E}[Y]^2},
\end{equation}

\begin{equation} \label{eq:appB5}
    \theta_{P_H} =  \frac{\mathbb{E}[Y^2]-\mathbb{E}[Y]^2}{\mathbb{E}[Y]-\frac{1}{1+e^{ab}}}.
\end{equation}

Plugging \eqref{eq:appB2} and
\eqref{eq:appB3} in \eqref{eq:appB4} and \eqref{eq:appB5}; \eqref{eq:kph} \eqref{eq:thetaph} are obtained and this completes the proof.

\section{Proof of Theorem \ref{theo:OPD1}} \label{App:proofOPD1}

After some manipulations, Eq. \eqref{eq:PrD1} can be written as

\begin{IEEEeqnarray}{rCl} \label{eq:appA1}
    P_{\rm out}^{D_1} \hspace{-0.1cm}=\hspace{-0.1cm} {\rm Pr}\hspace{-0.1cm}\left(\hspace{-0.1cm}X_1^2 \hspace{-0.1cm}<\hspace{-0.1cm}\max \left\{\hspace{-0.1cm} \frac{\xi_1(|h_{SI}|^2P_H\hspace{-0.1cm}+\hspace{-0.1cm}\sigma^2)}{\beta_{ss}\beta_{s1}},\hspace{-0.1cm} \frac{\xi_2(|h_{SI}|^2 P_H \hspace{-0.1cm}+\hspace{-0.1cm} \sigma^2)}{\beta_{ss}\beta_{s1}}\hspace{-0.1cm}\right\} \hspace{-0.1cm}\right)\hspace{-0.1cm}, \nonumber\\
\end{IEEEeqnarray}
where $\xi_2 = \frac{\gamma_{th2}}{P_t(1-\rho) (\alpha_2 - \alpha_1\gamma_{th2})}$ and $\xi_1=\frac{\gamma_{th}}{P_t(1-\rho)\alpha_1}$. Clearly, when $\rho=0$ or $\omega=0$, we have that

\begin{numcases}{P_{\rm out}^{D_1}=}
        \frac{\gamma\left(k_1, \frac{\xi_1 \sigma^2}{\theta_1\beta_{ss}\beta_{s1}} \right)}{\Gamma(k_1)}, & \qquad {\rm if } $\xi_1>\xi_2$ \nonumber
        \\
        \frac{\gamma\left(k_1,\frac{\xi_2 \sigma^2}{\theta_1\beta_{ss}\beta_{s1}}\right)}{\Gamma(k_1)}, & \qquad {\rm otherwise.}
\end{numcases}

When we have $\rho\neq0$ and $\omega\neq0$, \eqref{eq:appA1} should be further analyzed. Firstly, we should notice that according to Appendix \ref{app:proofPH},     $P_H$ is a Gamma r.v. with shape and scale parameters given by $k_{PH}$ and $\theta_{PH}$ respectively. Thus \eqref{eq:appA1} can be written as
\begin{IEEEeqnarray}{rCl} \label{eq:appA2}
    P_{\rm out}^{D_1} &=& 1 - {\rm Pr}\left( \underbrace{\xi_\ell|h_{ SI}|^2P_H}_{V} -\underbrace{X_1^2\beta_{ss}\beta_{s1}}_{Z}  <-  \xi_\ell\sigma^2  \right) \hspace{-0.15cm},
\end{IEEEeqnarray}
let us define $V=\xi_\ell|h_{ SI}|^2P_H$ and $Z=X_1^2\beta_{ss}\beta_{s1}$. Here, we proposed to approximate $V$ as a exponential r.v., $i.e.$, $V \overset{{\rm approx}}{\sim} {\rm Exponential}(\xi_\ell \omega k_{P_H}\theta_{P_H})$. Since $Z$ is the product of $X_1$ with a constant, we have $Z \sim {\rm Gamma}(k_1,\theta_1\beta_{ss}\beta_{s1})$, hence, \eqref{eq:appA2} can be written as 
\begin{IEEEeqnarray}{rCl}
    P_{\rm out}^{D_1} &=& 1 - \int_{\xi_\ell \sigma^2}^{\infty} F_{V}(z-\xi_\ell\sigma^2) f_Z(z) dz \nonumber
    \\
    \label{eq:AppBaux3}
    &=& 1 \hspace{-0.1cm} - \hspace{-0.1cm} \int_{\xi_\ell \sigma^2}^{\infty} f_Z(z) dz 
    +
    e^{\frac{\sigma^2}{\omega k_{P_H}\theta_{P_H}}} \hspace{-0.1cm} \int_{\xi_\ell \sigma^2}^{\infty} \hspace{-0.2cm} e^{-\frac{x}{\xi_1\omega k_{P_H}\theta_{P_H}}} f_Z(z) \nonumber
    \\
    &=& \frac{\gamma\left(k_1,\frac{\xi_{\ell} \sigma^2}{\theta_1\beta_{ss}\beta_{s1}}\right)}{\Gamma(k_1)} 
    +
    \frac{e^{\frac{\sigma^2}{\omega k_{P_H}\theta_{P_H}}}\xi_1 \omega k_{P_H} \theta_{P_H}}{\Gamma(k_1) } \nonumber
    \\
    &\times& \frac{\Gamma\left(k_1,\xi_1\sigma^2 \left(\frac{1}{\theta_1\beta_{ss} \beta_{s1}} + \frac{1}{\xi_1 \omega k_{P_H} \theta_{P_H}}\right) \right)}{\left(\xi_1 \omega k_{P_H} \theta_{P_H}+\theta_1\beta_{ss} \beta_{s1} \right)^{k_1}},
\end{IEEEeqnarray}
where we utilized the exponential CDF and \cite[3.381.9]{gradshteyn2014table} to solve \eqref{eq:AppBaux3}.

\section{Proof of Theorem \ref{theo:PoutD2}} \label{App:proofOPD2}

To derive the OP of $D_2$ it is reasonable to separate it in two cases, $i.e$,  when the $D_1$ does not operate as relay $(\rho = 0)$, and when the $D_1$ operates as a relay $(\rho \neq 0)$.

\subsection*{I) $\rho = 0$ ($D_1$ does not act as a relay)}

After some manipulations, we can written \eqref{eq:PrD2} as
\begin{IEEEeqnarray}{rCl}\label{eq:AppCaux1}
P_{\rm out}^{D_2} &=&  {\rm Pr}\left(|h_2|^2 < \frac{\sigma^2 \gamma_{th2} }{P_t(\alpha_2-\alpha_1\gamma_{th2})}    \right),
\end{IEEEeqnarray}
utilizing the CDF of exponential r.v., we obtain the following
\begin{equation}\label{eq:AppCaux2}
    P_{\rm out}^{D_2} = 1 - e^{- \frac{1}{\beta_{ss}\beta_{s2}N} \left(\frac{\sigma^2 \gamma_{th2} }{P_t(\alpha_2-\alpha_1\gamma_{th2})}\right)}.
\end{equation}

\subsection*{II) \underline{$\rho \neq 0$}
($D_1$ operates as a relay)}

After some manipulations, we can written \eqref{eq:PrD2} as

\begin{IEEEeqnarray}{rCl} \label{eq:AppCaux3}
    P_{\rm out}^{D_2}&=& {\rm Pr} \hspace{-0.1cm} \left(\hspace{-0.1cm}|h_1|^2\hspace{-0.2cm} <  \xi_2(\omega P_H+\sigma^2),|h_2|^2 < \frac{\sigma^2 \gamma_{th2} }{P_t(\alpha_2-\alpha_1\gamma_{th2})} \right) \nonumber
    \\ 
    &+& {\rm Pr}\left(|h_1|^2 \geq  \xi_2(\omega P_H+\sigma^2), \right. \nonumber 
    \\
     && \left. \frac{P_t |h_2|^2}{\sigma^2} < \frac{\gamma_{th2} - \frac{P_H}{\sigma^2} |h_{1,2}|^2}{\alpha_2-\alpha_1\gamma_{th2}+\alpha_1\frac{P_H}{\sigma^2} |h_{1,2}|^2}    \right).
\end{IEEEeqnarray}

Let us define $\overline{W} = \frac{P_H}{\sigma^2} |h_{1,2}|^2$, since $P_H$ assume lower values than $|h_{1,2}|^2$, we approximate $\overline{W}$ as exponential r.v.,  $\overline{W}\sim {\rm exp}\left(\frac{\theta_{P_H}k_{P_H}\beta_{12}}{\sigma^2}\right)$, hence

\begin{IEEEeqnarray}{rCl}  \label{eq:AppBaux4}
    && {\Pr} \left(\frac{P_t |h_2|^2}{\sigma^2} < \frac{\gamma_{th2} - \overline{W}}{\alpha_2-\alpha_1\gamma_{th2}+\alpha_1\overline{W}}    \right) \nonumber
    \\
    &=& \int_{0}^{\gamma_{th2}} \hspace{-0.3cm}\int_{0}^{\frac{\gamma_{th2} - y}{\alpha_2-\alpha_1\gamma_{th2}+\alpha_1y} } f_{\frac{P_t |h_2|^2}{\sigma^2}}(x) f_{\overline{W}}(y) dx dy \nonumber
    \\
    &=& \hspace{-0.1cm}  \underbrace{F_{\overline{W}}(\gamma_{th2}) \hspace{-0.1cm}-\hspace{-0.2cm}
    \int_{0}^{\gamma_{th2}}  \hspace{-0.15cm} \frac{\sigma^2e^{\left(\frac{-\sigma^2}{P_t\beta_{ss}\beta_{s2}N}\right)\hspace{-0.1cm} \left(\frac{\gamma_{th2} - y}{\alpha_2\hspace{-0.1cm}-\alpha_1\hspace{-0.05cm}\gamma_{th2}\hspace{-0.05cm}+\hspace{-0.05cm}\alpha_1y} \right)\hspace{-0.05cm}-\hspace{-0.05cm}\frac{\sigma^2}{\beta_{12}\theta_{P_H}k_{P_H}} y}}{\beta_{12}\theta_{P_H}k_{P_H}} }_{I}  \nonumber
    \\
\end{IEEEeqnarray}

Since $D_2$ assume low rate values, and $\alpha_2 > \alpha_1$ due to implementation of NOMA, we have $\alpha_2 \gg \alpha_1 \gamma_{th2}$, thus, the integral in $I$ can be approximated as

\begin{IEEEeqnarray}{rCl} \label{eq:AppBaux5}
       && \frac{e^{\frac{-\gamma_{th2}\sigma^2}{P_t\beta_{ss}\beta_{s2}N(\alpha_2-\alpha_1\gamma_{th2})}}\sigma^2}{\beta_{12}\theta_{P_H}k_{P_H}} \hspace{-0.15cm} \nonumber \\
       &\times&\int_{0}^{\gamma_{th2}}   e^{ y \left( \frac{\sigma^2 }{P_{t}N\beta_{ss}\beta_{s2}(\alpha_2
       -\alpha_1\gamma_{th2})} -\frac{\sigma^2}{\beta_{12}\theta_{P_H}k_{P_H}}\right)} dy, 
\end{IEEEeqnarray}
whose solution can be found trivially. Substituting the solution of \eqref{eq:AppBaux5} in \eqref{eq:AppBaux4} we obtain

\begin{IEEEeqnarray}{rCl} \label{eq:AppBaux6}
   &&I \approx 1 
   - 
   e^{\frac{\gamma_{th2}\sigma^2}{\theta_{P_H}k_{P_H}\beta_{12}}} 
   -
   \frac{e^{\frac{-\gamma_{th2}\sigma^2}{P_t\beta_{ss}\beta_{s2}N(\alpha_2-\alpha_1\gamma_{th2})}}\sigma^2}{\beta_{12}\theta_{P_H}k_{P_H}} \nonumber 
   \\
   &&\times \frac{\left( e^{\frac{\gamma_{th2}\sigma^2}{P_t\beta_{ss}\beta_{s2}N(\alpha_2-\alpha_1\gamma_{th2})} - \frac{\gamma_{th2 \sigma^2}}{\beta_{12\theta_{P_H}k_{P_H}}}} - 1 \right)}{\frac{\sigma^2}{P_t\beta_{ss}\beta_{s2}N(\alpha_2-\alpha_1\gamma_{th2})}-\frac{\sigma^2}{\beta_{12}\theta_{P_H}k_{P_H}}},
\end{IEEEeqnarray}
utilizing the analytical results derived in Appendix \ref{App:proofOPD1} and utilizing \eqref{eq:AppBaux6}, \eqref{eq:PoutD2} is obtained and this complete the proof. 



\end{document}